\documentclass[11pt,twosided]{article}
\usepackage[sans]{dsfont}

\usepackage[british]{babel}
\usepackage{amsmath,amssymb}
\usepackage{amsfonts}
\usepackage{amsthm}
\usepackage[nospace]{cite}

\usepackage{mathrsfs}
\usepackage[mathscr]{euscript}
\usepackage{mathbbol}
\theoremstyle{plain}
\newtheorem{theorem}{Theorem}
\theoremstyle{definition}

\newtheorem{assumption}{Assumption}
\theoremstyle{remark}
\newtheorem{remark}{Remark}
\newcommand{\Tr}{\operatorname{Tr}}
\newcommand{\rmd}{\mathrm{d}}
\newcommand{\rmi}{\mathrm{i}}
\newcommand{\rme}{\mathrm{e}}
\newcommand{\openone}{\mathds1}
\newcommand{\norm}[1]{\left\Vert#1\right\Vert}
\newcommand{\abs}[1]{\left\vert#1\right\vert}

\newcommand{\Cbb}{\mathbb{C}}
\newcommand{\Pbb}{\mathbb{P}}
\newcommand{\Qbb}{\mathbb{Q}}

\newcommand{\Kcal}{\mathcal{K}}
\newcommand{\Lcal}{\mathcal{L}}

\newcommand{\Pcal}{\mathcal{P}}
\newcommand{\Qcal}{\mathcal{Q}}
 \newcommand{\Rcal}{\mathcal{R}}

\newcommand{\Vcal}{\mathcal{V}}

\newcommand{\Fscr}{\mathscr{F}}

\newcommand{\Hscr}{\mathscr{H}}

\newcommand{\Lscr}{\mathscr{L}}

\newcommand{\Sscr}{\mathscr{S}}
\newcommand{\Tscr}{\mathscr{T}}

 \newcommand{\EE}{\operatorname{\mathbb{E}}}

\newcommand{\RE}{\operatorname{Re}}

\title{STOCHASTIC SCHR\"ODINGER EQUATIONS AND MEMORY}

\author{A. BARCHIELLI, \\
Politecnico di Milano, Department of Mathematics\\
Also: Istituto Nazionale di Fisica Nucleare, Sezione di Milano.
\\ \\
P. DI TELLA,
\\
Friedrich-Schiller-Universit\"{a}t Jena, Institut f\"{u}r Stochastik.
\\ \\
C. PELLEGRINI,
\\
Universit\'e Paul Sabatier, Laboratoire de Statistique et Probabilit\'es.
\\ \\
F. PETRUCCIONE,
\\ University of KwaZulu-Natal,
\\
School of Physics and National Institute for Theoretical Physics.}

\begin{document}
\maketitle
\begin{abstract}
By starting from the stochastic Schr\"odinger equation and quantum trajectory
theory, we introduce memory effects by considering stochastic adapted coefficients.
As an example of a natural non-Markovian extension of the theory of white noise
quantum trajectories we use an Ornstein-Uhlenbeck coloured noise as the output
driving process. Under certain conditions a random Hamiltonian evolution is
recovered. Moreover, we show that our non-Markovian stochastic Schr\"odinger
equations unravel some master equations with memory kernels.
\end{abstract}



\section{Introduction}

The problem of how to describe the reduced dynamics of a quantum open system
interacting with an environment is very important \cite{Alicki,Book,GardinerZoller}.
More and more applications demand to treat dissipative effects, tendency to
equilibrium, decoherence,... or how to have more equilibrium states, survival of
coherences and entanglement,... in spite of the interaction with the external
environment. The open system dynamics is often described in terms of \textit{quantum
master equations} which give the time evolution of the density matrix of the small
system. When the Markov approximation is good (no memory effects) the situation is
well understood: if the generator of the dynamics has the ``Lindblad structure'',
then the dynamics sends statistical operators into statistical operators and it is
completely positive \cite{lindblad1,GK}.

However, for many new applications the Markovian approximation is not applicable.
Such a situation appears in several concrete physical models: strong coupled
systems, entanglement and correlation in the initial state, finite reservoirs...
This gives rise to the theory of \textit{non-Markovian quantum dynamics}, for which
does not exist a general theory, but many different approaches \cite{Diosi4,%
Diosi3,Diosi2,francesco1,Breuer3,GW4,GW5,GW2,Budini1,%
Budini3,Budini2,francesco3,Reb,Breuer1,Diosi1,GW3,VacB09,ClementFrancesco}.

Non Markovian reduced dynamics are usually obtained from the total dynamics of
system plus bath by projection operator techniques such as Nakajima-Zwanzig operator
technique, time-convolutionless operator technique \cite{Book,Breuer3}, correlated
projection operator or Lindblad rate equations \cite{Breuer1,Budini2}... These
techniques yield in principle exact master equations for the evolution of the
subsystem. For example Nakajima-Zwanzig technique gives rise to an
integro-differential equation with a memory kernel involving a retarded time
integration over the history of the small system. However, in most of the cases the
exact evolutions remain of formal interest: no analytic expression of the solution,
impossible to simulate... Usually, some approximations have to be used to obtain a
manageable description. But as soon as an approximation is done, the resulting
equation can violate the complete positivity property; let us stress that the
complete positivity (and even positivity) is a major question in non-Markovian
systems \cite{Reb,VacB09}.

The easiest way to preserve complete positivity is to introduce approximate or
phenomenological equations at the Hilbert space level, an approach which is useful
also for numerical simulations. We can say that in this way one is developing a
non-Markovian theory of \emph{unravelling} and of ``Quantum Monte Carlo methods''
\cite{Book,francesco1,CarmichaelBook,Diosi4,Diosi3,Diosi2,GW5}. In the Markovian
case such an approach is related to the so called stochastic Schr\"odinger equation,
quantum trajectory theory, measurements in continuous time. It provides wide
applications for optical quantum systems and description of experiments such as
\emph{photo-detection} or \emph{heterodyne/homodyne detection}
\cite{CarmichaelBook,Bar86PR,Bel88,BarB91,Wis96,Mil96,Hol01,BarGreg09}. In the
non-Markovian case, an active line of research concentrates on a similar
interpretation of non-Markovian unravelling. In this context, the question is more
involved (for example, when complete positivity is violated the answer is hopeless)
and remains an open problem. For the usual scheme of indirect quantum measurement it
has been shown that in general such an interpretation is not accessible
\cite{GW3,ClementFrancesco}. Actually, only few positive answers for very special
cases have been found and this question is still highly debated
\cite{GW3,GW5,GW2,Diosi1}.

Our aim is to introduce memory at Hilbert space level, in order to guarantee at the
end a completely positive dynamics, and to maintain the possibility of the
measurement interpretation. Our approach is based on the introduction of stochastic
coefficients depending on the past history and on the use of coloured noises
\cite{BarH95,DiT,BarPP}.

In Sect.\ \ref{sec2} we introduce a special case of stochastic Schr\"odinger
equation with memory. The starting point is a generalisation of the usual theory of
the linear stochastic Schr\"odinger equation \cite{BarGreg09,DiT,BarH95}, based upon
the introduction of random coefficients. This approach introduces memory effects in
the underlying dynamics. The main interest is that the complete positivity is
preserved and a measurement interpretation can be developed. We present this theory
only in the context of the diffusive stochastic Schr\"odinger equation.

In Sect.\ \ref{sec:3} we attach the problem of memory by introducing an example of
\emph{coloured bath} and we show that we obtain a model of random Hamiltonian
evolution \cite{BarPP}, while we remain in the general framework of Sect.\
\ref{sec2}.

Finally, in Sect.\ \ref{sec:NZ}, by using Nakajima-Zwanzig projection techniques, we
show that the mean states satisfy closed master equations with memory kernels, which
automatically preserve complete positivity. Moreover, we can say that the stochastic
Schr\"odinger equations of the previous sections are unravellings of these memory
master equations.

\section{A non Markovian stochastic Schr\"odinger equation}\label{sec2}
The linear stochastic Schr\"odinger equation (lSSE) is the starting point to
construct unravelling of master equations and models of measurements in continuous
time \cite{BarB91,BarGreg09}. By introducing random coefficients in such equation,
but maintaining its structure, we get memory in the dynamical equations, while
complete positivity of the dynamical maps and the continuous measurement
interpretation are preserved \cite{BarH95,DiT}. To simplify the theory we consider
only diffusive contributions and bounded operators.

\begin{assumption}[The linear stochastic Schr\"odinger equation]\label{ass:lSSE}
Let $\Hscr$ be a complex, separable Hilbert space, the space of the quantum system,
and $\big(\Omega,\Fscr$, $(\Fscr_t),\mathbb{Q}\big)$ be a stochastic basis
satisfying the usual hypotheses, where a $d$-dimen\-sional continuous Wiener process
is defined; $\Qbb$ will play the role of a reference probability measure. The lSSE
we consider is
\begin{gather}\label{eq:lSSE}
\rmd \psi(t)= K(t) \psi(t)\rmd t +\sum_{j=1}^d R_j(t)\psi(t)\rmd W_j(t),
\\ \notag
\psi(0)=\psi_0\in L^2(\Omega,\Fscr_0,\Qbb;\Hscr).
\end{gather}
\end{assumption}

Let us denote by $\Tscr(\Hscr)$ the trace class on $\Hscr$, by $\Sscr(\Hscr)$ the
subset of the statistical operators and by $\Lscr(\Hscr)$ the space of the linear
bounded operators.

\begin{assumption}[The random coefficients] \label{ass:coeff}
The coefficient in the drift has the structure
\begin{equation}
K(t)=-\rmi H(t)-
\frac 1 2 \sum_{j=1}^d
R_j(t)^*R_j(t).
\end{equation}
The coefficients $H(t)$, $R_j(t)$ are random bounded operators with $H(t)=H(t)^*$,
say predictable c\`agl\`ad processes in $\big(\Omega, \Fscr, (\Fscr_t), \Qbb\big)$.

Moreover,  $\forall T
> 0$, we have
\begin{subequations}
\begin{gather}\label{ass:H}
\int_0^T\EE_\Qbb\left[\left\| H(t)\right\|\right]\rmd t <+\infty , \\ \label{ass:R}
\EE_\Qbb\biggl[\exp\biggl\{2\sum_{j=1}^d\int_0^T\left\| R_j(t)\right\|^2\rmd t
\biggr\}\biggr]<+\infty.
\end{gather}
\end{subequations}
\end{assumption}

\begin{theorem} \label{theor1}
Under Assumptions \ref{ass:lSSE}, \ref{ass:coeff}, the lSSE \eqref{eq:lSSE} has a
pathwise unique solution. The square norm $\norm{\psi(t)}^2$ is a continuous
positive martingale given by
\begin{equation}\label{normpsi}
\norm{\psi(t)}^2=\norm{\psi_0}^2 \exp\biggl\{\sum_j\biggl[\int_0^t m_j(s) \rmd W_j(s) - \frac 1 2
\int_0^t  m_j(s)^2\, \rmd s\biggr]\biggr\},
\end{equation}
\begin{equation}\label{mt}
m_j(t):=2\RE \left\langle \hat\psi(t)\big|R_j(t)\hat\psi(t)\right\rangle,
\end{equation}
\begin{equation}\label{hat_psi}
\hat\psi(t):= \begin{cases} \psi(t)/\norm{\psi(t)}, & \mathrm{if} \ \norm{\psi(t)}\neq 0,
\\
v \ \mathrm{(fixed\ unit\ vector)} ,
& \mathrm{if} \ \norm{\psi(t)}= 0.\end{cases}
\end{equation}
\end{theorem}
\begin{proof}
Assumptions \ref{ass:lSSE}, \ref{ass:coeff} imply the Hypotheses of
\cite[Proposition 2.1 and Theorem 2.4]{BarH95}, but Hypothesis 2.3.A of page 295.
According to the discussion at the end of p. 297, this last hypothesis can be
substituted by \eqref{ass:R}, which implies Novikov condition, a sufficient
condition for an exponential supermartingale to be a martingale. Then, all the
statements hold.
\end{proof}

\begin{remark}\label{rem1}
By expression \eqref{normpsi} we get that on the set $\{\norm{\psi_0}>0\}$ we have
$\norm{\psi(t)}>0$ $\Qbb$-a.s. This means that, if $\psi_0\neq 0$ $\Qbb$-a.s., then
the process $\hat\psi(t)$ \eqref{hat_psi} is almost surely defined by the
normalisation of $\psi(t)$ and the arbitrary vector $v$ does not play any role with
probability one.
\end{remark}

\begin{remark} \label{rem2}
Let us define the positive, $\Tscr(\Hscr)$-valued process
\begin{equation}\label{sigmat}
\sigma(t):=|\psi(t)\rangle \langle \psi(t)|.
\end{equation}

By applying the It\^o formula to $ \langle \psi(t)|a\psi(t)\rangle$,
$a\in\Lscr(\Hscr)$, we get the weak-sense linear stochastic master equation (lSME)
\begin{equation}\label{SME}
\rmd \sigma(t)= \mathcal{L}(t)[\sigma(t)]\rmd t + \sum_{j=1}^d \Rcal_j(t)
[\sigma(t)] \rmd W_j(t),
\end{equation}
\begin{equation}\label{Rcal}
\Rcal_j(t)[\rho]:=R_j(t)\rho+ \rho R_j(t)^*,
\end{equation}
\begin{equation}\label{rLop}
\mathcal{L}(t)[\rho]=  -\rmi [H(t),\rho] + \sum_{j=1}^d \left(
R_j(t)\rho R_j(t)^* - \frac 1 2 \left\{R_j(t)^* R_j(t) , \rho \right\} \right);
\end{equation}
$\Lcal(t)$ is the random Liouville operator \cite[Proposition 3.4]{BarH95}.
\end{remark}

\begin{assumption}[The initial condition] \label{ass:incond} Let us assume that the initial
condition $\psi_0$ is normalised, in the sense that
$\EE_\Qbb\left[\norm{\psi_0}^2\right]=1$. Then, $\varrho_0 :=\EE_\Qbb\left[|
\psi_0\rangle\langle\psi_0|\right]\in \Sscr(\Hscr)$ represents the initial
statistical operator.
\end{assumption}

\begin{remark}\label{rem3}
Under the previous assumptions $p(t):=\norm{\psi(t)}^2$ is a positive, mean-one
martingale and, $\forall T>0$, we can define the new probability law on $(\Omega,
\Fscr_T)$
\begin{equation}
\forall F\in\Fscr_T \qquad \Pbb^T(F):=\EE_\Qbb[p(T)1_F].
\end{equation}
By the martingale property these new probabilities are consistent in the sense that,
for $0\leq s <t$ and $F\in \Fscr_s$, we have $\Pbb^t(F)=\Pbb^s(F)$.
\end{remark}

The new probabilities are interpreted as the physical ones, the law of the output of
the time continuous measurement. Let us stress that it is possible to express the
physical probabilities in agreement with the axiomatic formulation of quantum
mechanics by introducing positive operator valued measures and completely positive
instruments \cite{BarH95,DiT}.

\begin{remark} \label{rem:Gir} By Girsanov theorem, the $d$-dimensional process
\begin{equation}\label{newW}
\widehat W_j(t):= W_j(t) -\int_0^t m_j(s)\,\rmd s, \qquad j=1,\ldots,d, \quad t\in[0,T],
\end{equation}
is a standard Wiener process under the physical probability $\Pbb^T$
\cite[Proposition 2.5, Remark 2.6]{BarH95}.
\end{remark}

By adding further sufficient conditions two more important equations can be
obtained.

\begin{assumption}[\!{\cite[Hypotheses 2.3.A]{BarH95}}]\label{ass:lSSE1}
Let us assume that we have
\begin{equation}
\sup_{\omega\in\Omega}\,\int_0^t\Bigl\|\sum_{j=1}^d R_j(s,\omega)^\ast
R_j(s,\omega)\Bigr\|\rmd s <+\infty .
\end{equation}
\end{assumption}

\begin{theorem} \label{theor2}
Let Assumptions \ref{ass:lSSE}--\ref{ass:lSSE1} hold. Under the physical probability
the normalized state $\hat\psi(t)$, introduced in Eq.\ \eqref{hat_psi}, satisfies
the non-linear \emph{stochastic Schr\"odinger equation} (SSE)
\begin{multline}\label{eq:nlSSE}
\rmd \hat\psi(t) = \sum_{j}\left[ R_j(t)-\RE  n_j\left(t,\hat \psi(t)\right)\right]
\hat\psi(t)\, \rmd  \widehat W_j (t) + K(t) \hat\psi(t)\,\rmd  t \\ {}+
\sum_j\biggl[ \left(\RE n_j\left(t,\hat \psi(t)\right)\right) R_j(t)- \frac 1 2
\left(\RE n_j\left(t,\hat \psi(t)\right)\right)^2\biggr] \hat\psi(t)\,\rmd t \,,
\end{multline}
where $ n_j(t,x):=\langle x|R_j(t)x\rangle$, $\forall t\in[0,+\infty)$, $j=1,
\ldots, d$, $ x\in\Hscr$.

Moreover, the process (\emph{a priori or average states}) defined by
\begin{equation}\label{etat}
\eta(t):=\EE_\Qbb[\sigma(t)], \quad \text{or} \quad \Tr \left\{a\eta(t)\right\}=\EE_\Qbb \left[
\langle \psi(t)|a\psi(t)\rangle \right] \quad \forall a \in \Lscr(\Hscr)
\end{equation}
satisfies the \emph{master equation}
\begin{equation}\label{eq:mas.eq}
\eta(t)=\varrho_0+\int_0^t\EE_\Qbb\big[\mathcal{L}(s)[\sigma(s)]\big]\rmd s.
\end{equation}
\end{theorem}
\begin{proof}
One can check that all the hypotheses of \cite[Theorem 2.7]{BarH95} hold. Then, the
SSE for $\hat \psi(t) $ follows.

As in the proof of \cite[Propositions 3.2]{BarH95}, one can prove that the
stochastic integral in the lSME \eqref{SME} has zero mean value. Then, Eq.\
\eqref{eq:mas.eq} follows.
\end{proof}

Note that, by the definition of the physical probabilities, we have also
\begin{equation}
\Tr \left\{a\eta(t)\right\}=\EE_{\Pbb^T} \left[
\langle \hat\psi(t)|a\hat\psi(t)\rangle \right], \quad \forall a \in \Lscr(\Hscr), \quad
\forall t,\,T : 0\leq t\leq T.
\end{equation}
The SSE \eqref{eq:nlSSE} is the starting point for numerical simulations; the key
point is that norm of its solution $\hat\psi(t)$ is constantly equal to one. We
underline that Eq.\ \eqref{eq:mas.eq} is not a closed equation for the mean state of
the system. In the last section we shall see how to obtain, in principle, a closed
equation for the \emph{a priori} states of the quantum system.

\subsection*{The finite dimensional case}
If we assume a finite dimensional Hilbert space and we strengthen the conditions on
the coefficients, we obtain a more rich theory. We discuss here below the situation
\cite{DiT}.

\begin{assumption}\label{ass:fin.dim}
\begin{enumerate}
\item The Hilbert space of the quantum system is finite dimensional, say
    $\Hscr:=\Cbb^n$. We write $M_n(\Cbb)$ for the space of the linear operators
    on $\Hscr$ into itself ($n\times n$ complex matrices).
\item The coefficient processes $R_j$ and $H$ are $M_n(\Cbb)$-valued and
    progressive with respect to the reference filtration.
\item The coefficients satisfy the following conditions: for every $T>0$ there
    exist two positive constants $M(T)$ and $L(T)$ such that
\begin{equation}
\begin{split}
&\sup_{\omega\in\Omega}\sup_{t\in[0,T]}\left\|\sum_{j=1}^d R_j(t,\omega)^\ast
R_j(t,\omega)\right\|
\leq L(T) <\infty ,
\\
&\sup_{\omega\in\Omega}\sup_{t\in[0,T]}\left\|H(t,\omega)\right\|\leq M(T) <\infty.
\end{split}
\end{equation}
\end{enumerate}
\end{assumption}

Under Assumption \ref{ass:fin.dim} it is possible to prove existence and pathwise
uniqueness of the solution of the lSSE just modifying classical results for
existence and uniqueness of the solution for stochastic differential equation with
deterministic coefficients.

Moreover, it is possible to prove that the solution of lSSE fulfil some $L^p$
estimate: in this point the finite dimension of the Hilbert space plays a
fundamental role because the bounds we obtain for the process $\psi(t)$ involve
constants depending on $n$.

Obviously, in this context the martingale property of the norm of the solution is
still valid and so one can define the consistent family of physical probabilities.
It is also possible to introduce the propagator of the lSSE, that is the two times
$M_n(\Cbb)$-valued stochastic process $A(t,s)$ such that $A(t,s)\psi(s)=\psi(t)$,
for all $t,s\geq0$ s.t. $s\leq t$. We are able to obtain a stochastic differential
equation (with pathwise unique solution) for the propagator and, by means of it, to
prove that the propagator takes almost surely values in the space of the invertible
matrices. The $L^p$ estimates for $\psi(t)$ are useful to obtain $L^p$ estimates on
$A(t,s)$. Furthermore, the propagator satisfies the typical composition law of an
evolution: $A(t,s)=A(t,r)A(r,s)$ for all $t,r,s\geq0$ s.t. $s\leq r\leq t$.

The almost sure invertibility of the propagator guarantees that the process
$\hat\psi(t)$ can be almost surely defined and that this process satisfies, under
the physical probabilities, a non linear SSE, similar to Eq.\ \eqref{eq:nlSSE}. It
is possible to prove that in this case the SSE has a pathwise unique solution.

When we go on extending the theory to the space of the statistical operators, we can
take as initial condition a random statistical operator or a deterministic one. We
define the process $\sigma(t)$ as in Eq.\ \eqref{sigmat} and, by using the It\^o
formula, we obtain an equation formally similar to the lSME, but in this case we are
able to prove the uniqueness of its solution given the initial statistical operator
$\varrho_0$ (the existence comes out by construction). In this way we can say that
the lSME is the evolution equation of the quantum system, when the initial condition
is a deterministic (or even random) statistical operator $\varrho_0$. We can
introduce the propagator of the lSME, which is a two times-linear map valued
stochastic process, say $\Lambda(t,s)$, such that
$\Lambda(t,0)[\varrho_0]=\sigma(t)$, $\Lambda(t,s)=\Lambda(t,r)\circ\Lambda(r,s)$,
for all $t,r,s\geq0$ s.t. $s\leq r\leq t$ and $\Lambda(t,s)[\tau]=A(t,s)\tau
A(t,s)^*$. From the last expression of the propagator of the lSME, it comes out that
this is a completely positive map valued process.

Also in this case we can introduce a consistent family of physical probabilities.
Indeed, the process $\Tr\{\sigma(t)\}$ is an exponential mean-one martingale that
can be used to define the new probability laws, as we did in the Hilbert space.

It is then possible to define the normalisation of $\sigma(t)$ with respect to its
trace,
\[\varrho(t)=\frac{\sigma(t)}{\Tr\{\sigma(t)\}}\]
and, under the physical probabilities, we have the following non linear equation for
$\varrho(t)$, with pathwise unique solution
\begin{equation}\label{eq:nlSME}
\begin{cases}
\displaystyle\rmd \varrho(t)=\Lcal(t)[\varrho(t)]\rmd t +\sum_{j=1}^d \left\{\Rcal_j(t)[\varrho(t)]-
 v_j(t)\varrho(t)\right\}\rmd \widehat W_j(t)\,,&\quad t\geq0
\\
\varrho(0)=\varrho_0\,,
\end{cases}
\end{equation}
where $v_j(t):=\Tr\{(R_j(t)+R_j(t)^*)\varrho(t)\}$, and $\widehat W(t)$ is a Wiener
process under the physical probabilities defined by
\[\widehat W_j(t):=W_j(t)-\int_0^tv_j(s)\rmd s\,,\quad\forall j=1,\ldots,d\,.\]

\section{Random Hamiltonian}\label{sec:3}
In the previous section, we have presented a non Markovian generalisation of the
usual diffusive lSSE by using random coefficients to introduce memory. In this
section we adopt an alternative strategy and we start with a usual lSSE with non
random coefficients, but driven by a coloured noise; in this way the memory is
encoded in the driving noise of the lSSE, not in the coefficients. As we shall see,
this model too turns out to be a particular case of the general theory presented in
Section \ref{sec2}. Moreover, the new lSSE will be norm-preserving and will
represent a quantum system evolving under a random Hamiltonian dynamics, while the
Hamiltonian is very singular and produces dissipation.

As our aim is just to explore some possibility, we keep things simple and we
consider a one-dimensional driving noise and two non-random, bounded operators $A$
and $B$ on $\Hscr$ in the drift and in the diffusive terms. The starting point is
then the basic linear stochastic Schr\"odinger equation
\begin{equation}\label{LIN1}
\rmd\psi(t)=A\psi(t)\rmd t+B\psi(t)\rmd X(t).
\end{equation}

The simplest choice of a coloured noise is the stationary Ornstein-Uhlen\-beck
process defined by
\[
X(t) = \rme^{-\gamma t} Z+\int_0^t \rme^{-\gamma (t-s)} \rmd W(s), \qquad \gamma>0.
\]
where $(W(t))$ is a one dimensional Wiener process, defined on the stochastic basis
$(\Omega,\Fscr,\Fscr_t,\mathbb{Q})$ and $Z$ is an $\Fscr_0$-measurable, normal
random variable with mean 0 and variance $1/(2\gamma)$. The Ornstein-Uhlenbeck
process $(X(t))$ is a Gaussian process with zero mean and correlation function
\begin{equation}
\EE_\Qbb [X(t)X(s)]=\frac{\rme^{-\gamma\abs{t-s}}}{2\gamma}\,;
\end{equation}
it satisfies the stochastic differential equation
\begin{equation}\label{eqor}
\rmd X(t)= -\gamma X(t) \rmd t +
\rmd W(t), \qquad X(0)=Z.
\end{equation}
Formally, Eq.\ \eqref{LIN1} is driven by the derivative of the Ornstein-Uhlenbeck
process (heuristically, $\rmd X(t)=\dot X(t)\rmd t$), whose two-time correlation is
no more a delta, as in the case of white noise, but it is formally given by
$\EE_\mathbb{Q}[\dot X(t)\dot X(s)]=\delta(t-s) - \frac \gamma 2
\,\rme^{-\gamma\vert t-s\vert}$. Note that the Markovian regime is recovered in the
limit $\gamma \downarrow 0$.

It is then straightforward that Eq. (\ref{LIN1}) can be rewritten in the form
\begin{equation}\label{lin2}
\rmd \psi(t)=\bigl(A-\gamma X(t)B\bigr)\psi(t)\rmd t+ B\psi(t)\rmd W(t),
\end{equation}
on $(\Omega,\Fscr,\Fscr_t,\mathbb{Q})$. The initial condition is assumed to satisfy
Assumption \ref{ass:incond}. Assumption \ref{ass:lSSE} is satisfied with $d=1$,
$K(t)=A-\gamma X(t)B$, $R(t)=B$.

The key point of the construction of Section \ref{sec2} and of its interpretation is
the fact that $\norm{\psi(t)}^2$ is a martingale. To this end we compute its
stochastic differential by using the It\^o rules and we get
\begin{multline}\label{diffnormsq}
\rmd\langle\psi(t)|\psi(t)\rangle=\langle
\rmd\psi(t)|\psi(t)\rangle+\langle\psi(t)|\rmd\psi(t)\rangle+
\langle \rmd\psi(t)|\rmd\psi(t)\rangle \\
{}=\langle\psi(t)|\left[A^*+A-\gamma X(t)\left(B^*+B\right)+ B^*
B\right]\psi(t)\rangle\rmd t \\
{}+\langle\psi(t)|\big( B^*+ B\big)\psi(t)\rangle\rmd W(t).
\end{multline}
Then, the process $(\Vert\psi(t)\Vert^2)$ can be a martingale only if the term in
front of $\rmd t$ (the drift term) is equal to zero. This imposes that
\begin{equation}\label{cond1}
A^*+A-\gamma X(t)\left(B^*+B\right)+B^* B=0,\qquad \forall t.
\end{equation}
By taking the mean of this equation we get $A^*+A+B^* B=0$; then, we need also
$B^*+B$. These conditions impose that there are two self-adjoint operators $L$ and
$H_0$ such that $B=-\rmi L$ and $A=-\rmi H_0-\frac{1}{2}L^2$. As a consequence the
initial equation \eqref{LIN1} becomes
\begin{equation}\label{liin3}
\rmd\psi(t)=\left[-\rmi\left(H_0-\gamma X(t)L\right)-\frac{1}{2}\,L^2\right]\psi(t) \rmd t
-\rmi L\psi(t)\rmd W(t)\,.
\end{equation}

Now, being $X(t)$ a continuous adapted process and $H_0={H_0}^*\in \Lscr(\Hscr)$,
$L=L^*\in \Lscr(\Hscr)$, also Assumption \ref{ass:coeff} holds with
\[
H(t)=H_0-\gamma X(t)L, \quad R(t)=-\rmi L, \quad K(t)=-\rmi\bigl(H_0-\gamma
X(t)L\bigr)-\frac{1}{2}\,L^2.
\]
Moreover, we have $\EE_\Qbb[\abs{X(t)}]\leq \sqrt{\EE_\Qbb[X(t)^2]}\leq
1/\sqrt{2\gamma}$, \[ \EE_\Qbb[\norm{H(t)}]\leq \norm{H_0}+\gamma
\norm{L}\EE_\Qbb[\abs{X(t)}]\leq \norm{H_0}+\sqrt{\frac\gamma 2} \norm{L}, \] which
implies condition \eqref{ass:H}. Condition \eqref{ass:R} and Assumption
\ref{ass:lSSE1} are trivially satisfied because $R(t)$ is non random, time
independent and bounded.

As all Assumptions \ref{ass:lSSE}--\ref{ass:lSSE1} hold, also all statements of
Theorems \ref{theor1}, \ref{theor2} and Remarks \ref{rem1}--\ref{rem:Gir} hold. In
particular the lSSE \eqref{liin3} has a pathwise unique solution.

What is peculiar of the present model is that Eqs.\
\eqref{diffnormsq}--\eqref{liin3} give $\norm{\psi(t)}^2$ $=\norm{\psi(0)}^2$ or
that the probability densities are independent of time, $p(t)=p(0)$, cf.\ Eqs.\
\eqref{mt} and \eqref{newW}, which give $m(t)=0$ and $\widehat W(t)=W(t)$. We have
also, from Eq.\ \eqref{hat_psi}, $\hat\psi(t)= \psi(t)/\norm{\psi(0)}$, if
$\norm{\psi(0)}\neq 0$, and, from Remark \ref{rem3}, $\Pbb^t(F)=\Qbb(F)$, for all
events $F\in \Fscr_t$, independent of $\Fscr_0$. As a consequence the change of
probability has no effect (the new probability is equal to the initial for events
independent of $\Fscr_0$). In other terms, no information has been extracted from
the measurement interpretation.

Moreover, the property $\|\psi(t)\|=\|\psi_0\|$ is in agreement with a purely
Hamiltonian evolution. More precisely, let
$\displaystyle\mathop{\mathrm{T}}^\leftarrow\exp\{\cdots\}$ denotes the time ordered
exponential; then, the formal solution of Eq.\ \eqref{liin3} is given by
\[
\psi(t)=\mathop{\mathrm{T}}^\leftarrow\exp\left\{-\rmi \int_0^t\left(H_0-\gamma X(s)L\right) \rmd s
-\rmi \int_0^t L\,\rmd W(s) \right\}
\psi_0.
\]
The evolution of the quantum system is then completely determined by the
time-dependent, random Hamiltonian \[ \hat{H}_t=H_0+\left( \dot W(t)-\gamma X(t)
\right) L. \]
Let us stress that it is a formal expression, due to the presence of
$\dot W(t)$.

This shows that the usual measurement interpretation of \eqref{LIN1}
\textit{coloured} with an Ornstein-Uhlenbeck process gives raise to a random
Hamiltonian evolution. As announced, we recover the framework of the evolution of a
closed system incorporating a random environment characterised in terms of an
Ornstein-Uhlenbeck noise.

One can investigate the evolution of the corresponding density matrices. To this
end, we consider the pure state process $(\sigma(t))$ defined by $
\sigma(t)=|\psi(t)\rangle \langle \psi(t)|$, Eq.\ \eqref{sigmat}. By using It\^o
rules, the process $(\sigma(t))$ satisfies the stochastic differential equation
(SDE)
\begin{equation}\label{eqrho1}
\rmd \sigma(t)= -\rmi [H_0-\gamma X(t)L,\sigma(t)]\rmd t -\rmi  [L,\sigma(t)]\rmd W(t) -\frac {1}
2 \,[L,[L,\sigma(t)]]\rmd t,
\end{equation}
which is, of course, equivalent to \eqref{liin3} with random Liouville operator
\eqref{rLop} given by \[ \Lcal(t)= -\rmi [H_0-\gamma X(t)L,\cdot]\rmd t  -\frac {1}
2 \,[L,[L,\cdot]]\rmd t. \]
Let us stress that the presence of the Ornstein-Uhlenbeck
process implies that the solution $(\sigma(t))$ of Eq.\ \eqref{eqrho1} is not a
Markov process.

Taking the expectation, we get the evolution of the mean
$\eta(t)=\EE_\Qbb[\sigma(t)]$, Eq.\ \eqref{etat}, which turns out to be
\begin{equation}\label{meqmem}
\frac{\rmd\ }{\rmd t}\, \eta(t)= -\rmi [H_0,\eta(t)] -\frac {1} 2 \,[L,[L,\eta(t)]]
+\rmi \gamma\big[L,\EE_\Qbb[X(t)\sigma(t)]\big].
\end{equation}
Note that it is not a closed master equation for the mean state $\eta(t)$. Actually,
we have derived a model with memory for the mean state. Indeed, the term $\rmi\gamma
\big[L,\EE_\Qbb[X(t)\sigma(t)]\big]$ introduces non-Markovian memory effects in the
dynamics. Moreover, Eq.\ \eqref{liin3} is an unravelling of the master equation
\eqref{meqmem}.


\section{Projection techniques and closed master equations with memory}\label{sec:NZ}

As we have seen in Eq.\ \eqref{eq:mas.eq}, the \emph{a priori} states or average
states $\eta(t)=\EE_\Qbb[\sigma(t)]=\EE_{\Pbb^t}[\rho(t)]$ satisfy the equation
$\dot \eta(t)=\EE_\Qbb \big[ \Lcal(t)[\sigma(t)]\big]$, which is not closed because
both $\Lcal(t)$ and $\sigma(t)$ are random. However, at least heuristically, some
kind of generalised master equations can be obtained by using the Nakajima-Zwanzig
projection technique \cite[Section 9.1.2]{Book}.

Let us introduce the projection operators on the relevant part (the mean) and on the
non relevant one:
\[
\Pcal[\cdots]:= \EE_\Qbb[\cdots],\qquad \Qcal:=\openone - \Pcal.
\]
Then, we have $\eta(t)=\Pcal[\sigma(t)]$ and we define the non relevant part of the
state, the mean Liouville operator and the difference from the mean of the Liouville
operator
\[
\sigma_\bot(t):= \Qcal[\sigma(t)]=\sigma(t)-\eta(t), \]
\[ \Lcal_{\mathrm{M}}(t):= \EE_\Qbb[\Lcal(t)],
\qquad \Delta \Lscr(t):=\Lscr(t)-\Lcal_{\mathrm{M}}(t).
\]
By using the projection operators and the fact that the stochastic integrals have
zero mean, which means $\Pcal \int_0^t \rmd W_j(s) \cdots=0$, from \eqref{SME} we
get the system of equations
\begin{subequations}
\begin{equation}\label{doteta1}
\dot \eta(t)= \Lcal_{\mathrm{M}}(t)[\eta(t)]+
\Pcal\circ\Delta\mathcal{L}(t)[\sigma_\bot(t)],
\end{equation}
\begin{multline}\label{eq:sigmabot}
\rmd \sigma_\bot(t)= \Qcal\circ\mathcal{L}(t)[\sigma_\bot(t)] \rmd t + \sum_{j=1}^d
\Rcal_j(t)  [\sigma_\bot(t)] \rmd W_j(t)
\\ {} +
\Qcal\circ\mathcal{L}(t) [\eta(t)] \rmd t + \sum_{j=1}^d \Rcal_j(t) [\eta(t)] \rmd
W_j(t).
\end{multline}
\end{subequations}

As one can check by using It\^o formula, the formal solution of Eq.\
\eqref{eq:sigmabot} can be written as
\begin{multline}\label{bot2}
\sigma_\bot(t)= \Qcal\circ \Vcal(t,0)[\sigma_\bot(0)]+ \int_0^t \Qcal\circ
\Vcal(t,s)\circ\biggl( \mathcal{L}(s) -\sum_j \Rcal_j(s)^2\biggr) [\eta(s)] \rmd s
\\ {} +
\Qcal\circ \Vcal(t,0)\biggl[\sum_{j=1}^d  \int_0^t \Vcal(s,0)^{-1}\circ \Rcal_j(s)
[\eta(s)] \rmd W_j(s)\biggr],
\end{multline}
where $\Vcal(t,r)$ is the fundamental solution (or propagator) of the lSME
\eqref{SME} and satisfies the SDE
\[
\Vcal(t,r)=\openone+ \int_r^t\rmd s\, \Lcal(s)  \circ\Vcal(s,r)
+\sum_{j=1}^d \int_r^t \rmd W_j(s)\,\Rcal_j(s)\circ \Vcal(s,r).
\]
Let us stress that if one includes $\Vcal(t,0)$ into the stochastic integrals, from
one side one gets the simpler expression $\Vcal(t,0)\circ
\Vcal(s,0)^{-1}=\Vcal(t,s)$. But the propagator is a stochastic process and could be
non adapted. To overcome this difficulty one should use some definition of
anticipating stochastic integral. So, we prefer the formulation with $\Vcal(t,0)$
outside the stochastic integral, in order to have only adapted integrands.

By introducing the quantity \eqref{bot2} into Eq.\ \eqref{doteta1}, we get the
generalised master equation for the \emph{a priori} states
\begin{multline}\label{MME}
\dot \eta(t) = J(t)+ \Lcal_{\mathrm{M}}(t)[\eta(t)]+ \int_0^t
\Kcal(t,s)[\eta(s)]\rmd s
\\ {} +
\EE_\Qbb\biggl[\Delta\Lcal(t)\circ\Qcal\circ \Vcal(t,0)\biggl[\sum_{j=1}^d \int_0^t
\Vcal(s,0)^{-1}\circ \Rcal_j(s) [\eta(s)] \rmd W_j(s)\biggr]\biggr],
\end{multline}
where
\[
J(t):=\EE_\Qbb\left[\Delta\Lcal(t)\circ\Qcal\circ
\Vcal(t,0)[\sigma_\bot(0)]\right]
\]
is an inhomogeneous term which disappears if the initial state is non random, i.e.
$\sigma_\bot(0)=0$, and
\[
\Kcal(t,s):=\EE_\Qbb\bigg[\Delta\Lcal(t)\circ \Qcal\circ\Vcal(t,s)\circ\biggl(
\mathcal{L}(s) -\sum_j \Rcal_j(s)^2\biggr)\bigg]
\]
is an integral memory kernel. Also the last term in \eqref{MME} is a memory
contribution, linear in $\eta$ and depending on its whole trajectory up to $t$.

Let us stress that to compute the terms appearing in Eq.\ \eqref{MME} and to solve
it is not simpler than to solve Eq.\ \eqref{SME} and to compute the mean of the
solution. The meaning of Eq.\ \eqref{MME} is theoretical: it is a quantum master
equation with memory and Eq.\ \eqref{eq:nlSSE} gives an unravelling of it.

While the best way to study a concrete model is to simulate the stochastic
Schr\"odinger equation, Eq.\ \eqref{MME} could be the starting point for some
approximation. A possibility is to approximate $\Vcal(t,r)$ by the deterministic
evolution generated by the mean Liouville operator:
\[
\Vcal_{\mathrm{M}}(t,r)=\openone+ \int_r^t\rmd s\, \Lcal_{\mathrm{M}}(s)  \circ\Vcal_{\mathrm{M}}(s,r).
\]
If we take also $\sigma_\bot(0)=0$, we get
\begin{multline}\label{MMEx}
\dot \eta(t) \simeq \Lcal_{\mathrm{M}}(t)[\eta(t)]+ \int_0^t
\Kcal_1(t,s)[\eta(s)]\rmd s
\\ {} +
\EE_\Qbb\biggl[\Delta\Lcal(t)\biggl[\sum_{j=1}^d \int_0^t
\Vcal_{\mathrm{M}}(t,s)\circ \Rcal_j(s) [\eta(s)] \rmd W_j(s)\biggr]\biggr],
\end{multline}
\[
\Kcal_1(t,s):=\EE_\Qbb\bigg[\Delta\Lcal(t)\circ\Vcal_{\mathrm{M}}(t,s)\circ\biggl(
\Delta\mathcal{L}(s) -\sum_j \Delta\Rcal_j^2(s)\biggr)\bigg],
\]
where $\Delta\Rcal_j^2(s)=\Rcal_j(s)^2-\EE_\Qbb[\Rcal_j(s)^2]$.

For the model of the previous section we have: $\Rcal_j(s)=\Rcal=-\rmi [L,\cdot]$,
\[
\Lcal_{\mathrm{M}}(t)=\Lcal_{\mathrm{M}}=-\rmi [H_0, \cdot]-\frac 1 2 [L,[L, \cdot]], \qquad
\Vcal_{\mathrm{M}}(t,s)=\rme^{\Lcal_{\mathrm{M}}\left(t-s\right)},
\]
\begin{multline*}
\EE_\Qbb\biggl[\Delta\Lcal(t)\biggl[\sum_{j=1}^d \int_0^t
\Vcal_{\mathrm{M}}(t,s)\circ \Rcal_j(s) [\eta(s)] \rmd W_j(s)\biggr]\biggr]
\\ {}=
-\gamma \EE_\Qbb \bigg[X(t)\int_0^t\rmd W(s) \, \Rcal\circ
\rme^{\Lcal_{\mathrm{M}}\left(t-s\right)} \circ \Rcal[\eta(s)]\bigg]
\\ {}=
-\gamma \int_0^t\rmd s \, \Rcal\circ \rme^{\left(\Lcal_{\mathrm{M}}-
\gamma\right)\left(t-s\right)} \circ \Rcal[\eta(s)],
\end{multline*}
\[
\Kcal_1(t,s)=\gamma^2\EE_\Qbb[X(t)X(s)]\Rcal\circ
\rme^{\Lcal_{\mathrm{M}}\left(t-s\right)} \circ \Rcal  =\frac \gamma 2 \,
\Rcal\circ \rme^{\left(\Lcal_{\mathrm{M}}-\gamma\right)\left(t-s\right)} \circ
\Rcal.
\]
Finally, the approximation of the non Markovian master equation turns out to be
\begin{multline}
\dot \eta(t) \simeq -\rmi [H_0, \eta(t)]-\frac 1 2 [L,[L, \eta(t)]]
\\ {}
+ \frac \gamma 2 \int_0^t \rmd s \left[ L , \rme^{\left(\Lcal_{\mathrm{M}}-
\gamma\right)\left(t-s\right)}\bigr[ \left[L,\eta(s)\right]\bigr]\right].
\end{multline}
However, we have no results on the positivity preserving properties of such an
approximate evolution equation, while the unravelling of the complete equation
guarantees complete positivity and feasibility of numerical simulations.

\section*{Acknowledgments}

PDT thanks his PhD advisor, Prof.\ Dr.\ Hans-Jurgen Engelbert, and acknowledges the
financial support of the Marie Curie Initial Training Network (ITN),
FP7-PEOPLE-2007-1-1-ITN, no.213841-2, Deterministic and Stochastic Controlled
Systems and Applications.

CP acknowledges the financial support of the ANR ``Hamiltonian and Markovian
Approach of Statistical Quantum Physics'' (A.N.R. BLANC no ANR-09-BLAN-0098-01).

\end{document}